\documentclass[10pt, centerh1]{article}
\usepackage{amssymb,amsmath,colonequals,color,multirow}

\usepackage{natbib}
\usepackage{setspace}
\usepackage{graphicx}
\usepackage{diagbox}
\usepackage{colortbl,xcolor}

\usepackage{amsthm,amsmath}
\usepackage{graphicx}

\newtheorem{theorem}{Theorem}

\usepackage{natbib}

\usepackage{amsmath}  {
      \theoremstyle{plain}
      \newtheorem{assumption}{Assumption}
 }
   
 \theoremstyle{remark}
\newtheorem*{remark}{Remark}

\begin{document}
\title{Statistical inference for sketching algorithms}
\author{Ryan P.\ Browne\thanks{Department of Statistics \& Actuarial Science, University of Waterloo, Waterloo, Ontario, N1E 2V1, Canada. Email: rpbrowne@uwaterloo.ca}  \ and Jeffrey L.\ Andrews\thanks{Department of Statistics, University of British Columbia, Okanagan Campus, Kelowna, BC, V1V 1V7, Canada. Email: jeff.andrews@ubc.ca}}


\maketitle

\begin{abstract} 
Sketching algorithms use random projections to generate a smaller sketched data set, often for the purposes of modelling. Complete and partial sketch regression estimates can be constructed using information from only the sketched data set or a combination of the full and sketched data sets. Previous work has obtained the distribution of these estimators under repeated sketching, along with the first two moments for both estimators. Using a different approach, we also derive the distribution of the complete sketch estimator, but additionally consider the error term under both repeated sketching and sampling. Importantly, we obtain pivotal quantities which are based solely on the sketched data set --- specifically not requiring information from the full data model fit. These pivotal quantities can be used for inference on the full data set regression estimates or the model parameters. For partial sketching, we derive pivotal quantities for a marginal test and an approximate distribution for the partial sketch under repeated sketching or repeated sampling, again avoiding reliance on a full data model fit. We extend these results to include the Hadamard and Clarkson--Woodruff sketches then compare them in a simulation study. 
\end{abstract}

\section{Introduction}
\label{sec:intro}

Sketching algorithms use random projections to reduce a large data set to a smaller sketched data set. This sketched data set can be used as a surrogate for the full data in cases where computation on a massive data set could be computationally infeasible --- even those as simple as summary statistics. In a popular application of sketching, one can calculate the regression estimates on a sketched data set which can be used as a surrogate to the regression estimates that would have arrived via fitting on the full data set. \cite{woodruff14} and \cite{mahoney16} quantify the precision of this surrogate estimate using probabilistic worst-case bounds. Alternatively, \cite{ahfock20} derived distributions and moments for these estimators. For further overview on sketching see \cite{mahoney11} and \cite{cormode11}; for an overview in the context of regression see \cite{ahfock20}.


In regression,  we suppose the full data set consists of $n \times p$ covariate matrix, $ \mathbf{X}$, and  $n$ dimensional covariate vector, $ \mathbf{Y}$ where $n >p$ and $ \mathbf{X}$ is full rank. The Gaussian regression model assumes that the response was generated via  $ \mathbf{Y} \sim N( \mathbf{X} \boldsymbol{\beta}_0, \sigma^2 I)$ where $\boldsymbol{\beta}_0$ is a $p$ dimensional vector. The least squares estimate of $\boldsymbol{\beta}_0$ using the full data set is 
\begin{equation*}
\boldsymbol{\beta}_F = \left( \mathbf{X}^{\mathsf{T}} \mathbf{X} \right)^{-1} \mathbf{X}^{\mathsf{T}} \mathbf{y}.
\end{equation*}
For inference two other important quantities are the residual sum of squares and the model sum of squares,
\begin{equation*}
SSR_F =  \mathbf{y}^{\mathsf{T}}  \{ \mathbf{I} - \mathbf{X} \left(\mathbf{X}^{\mathsf{T}} \mathbf{X} \right)^{-1} \mathbf{X}^{\mathsf{T}} \} \mathbf{y},
\quad
SSM_F =  \mathbf{y}^{\mathsf{T}} \mathbf{X} \left(\mathbf{X}^{\mathsf{T}} \mathbf{X} \right)^{-1} \mathbf{X}^{\mathsf{T}}  \mathbf{y}. 
\end{equation*}
The $SSR_F/(n-p)$ is the least squares or the unbiased estimate of the model error $\sigma^2$.

Sketching algorithms uses random projections to reduce the size of the data set from $n$ to $k$ observations. These random projections are  represented as a $k \times n$ sketching matrix $\mathbf{S}$. There are many methods to construct a $\mathbf{S}$, for example the Gaussian, Hadamard  \citep{ailon09} or Clarkson--Woodruff sketch \citep{clarkson17} approaches.
We focus, initially, on the Gaussian sketch matrix in which elements are generated by $\mathbf{S}_{ij} \sim N(0, 1/k)$.  Then the Gaussian sketch is a $k\times n$ random matrix $\mathbf{S}$ such that
\begin{equation*}
vec S
\sim N_{kn} \left(  \mathbf{0}_{kn}, \mathbf{I}_n \otimes \frac{1}{k} \mathbf{I}_k \right)
\end{equation*}
where the $vec$ operator converts a matrix into a column vector. The sketched covariate matrix is $\mathbf{X}_s = \mathbf{S} \mathbf{X}$ while the sketched response vector 
is $\mathbf{y}_s = \mathbf{S} \mathbf{y}$, and therefore  
the complete sketching estimate is 
\begin{equation*} 
\boldsymbol{\beta}_s = \left( \mathbf{X}_{s}^{\mathsf{T}} \mathbf{X}_s \right)^{-1} \mathbf{X}_{s}^{\mathsf{T}} \mathbf{y}_s . 
\end{equation*}
\cite{ahfock20} derived the distribution of the estimator under repeated sketching,
 \begin{equation} \label{complete sketch dist}
\boldsymbol{\beta}_s  \mid \mathbf{X}, \mathbf{y} 
\sim t_{p}\left[ k-p+1,   \boldsymbol{\beta}_{F},
\left( \mathbf{X}^{\mathsf{T}} \mathbf{X} \right)^{-1} \frac{SSR_F}{ k-p+1}  \right].
 \end{equation}
In addition, \cite{ahfock20}  provides a sketching central limit theorem which extends this result to sketching projections beyond the Gaussian sketch to the Hadamard and Clarkson--Woodruff sketch. However a claim from  \cite{ahfock20}  is that ``an immediate consequence of [(\ref{complete sketch dist})] is the ability to generate exact confidence intervals for the elements of '' $\boldsymbol{\beta}_F$. It's important to stress that this result specifically depends on $SSR_F$, which are computed via a full model fit, and thereby require computation of $\boldsymbol{\beta}_F$ --- exactly the issue sketching is trying to resolve. In Section~\ref{sec:complete}, we build on work from \cite{ahfock20} by further investigating the distribution of the sketched residuals and provide pivotal quantities which are computed solely from the sketched data set.

We also further explore inference for partial sketching in Section~\ref{sec:part}. The partial sketch uses information from both the sketched data set and the full data set. Similar to \cite{ahfock20}, we only consider the unbiased partial sketch  denoted 
\begin{equation*}
\boldsymbol{\beta}_p  
= \gamma \times \left( \mathbf{X}_{s}^{\mathsf{T}} \mathbf{X}_s \right)^{-1} \mathbf{X}^{\mathsf{T}} \mathbf{y}
\end{equation*}
where $\gamma =(k-p-1)/k$. The partial sketch needs to be adjusted because  $( \mathbf{X}_{s}^{\mathsf{T}} \mathbf{X}_s )^{-1}$ follows the inverse Wishart distribution with degrees of freedom equal to $k-p-1$ and matrix parameter equal to $k (\mathbf{X}^{\mathsf{T}} \mathbf{X})^{-1}$.
\cite{ahfock20} show that the variance of the of unbiased partial sketch, $var( \boldsymbol{\beta}_p )$ is 
 \begin{equation*}
\frac{  (k-p-1) }{  \left( k -p \right) \left( k-p -3 \right)  } \left\{ \mathbf{y}^{\mathsf{T}}  \mathbf{X}^{\mathsf{T}} \left(\mathbf{X}^{\mathsf{T}}  \mathbf{X} \right)^{-1}  \mathbf{X}  \mathbf{y}   \times   \left(\mathbf{X}^{\mathsf{T}}  \mathbf{X} \right)^{-1} 
+ \frac{k-p+1}{k-p-1} \times  \boldsymbol{\beta}_{F}  \boldsymbol{\beta}_{F}^{\mathsf{T}} \right\}
\end{equation*}
and also give a formula for the unconditional variance. In Section~\ref{sec:part} we seek  distributions with pivots unreliant on the full data model fit, finding approximate distributional forms that can be used for inference in partial sketching regression. 

Finally, we extend these results for complete and partial sketching to include the Hadamard and Clarkson--Woodruff sketches. These results are based on methodology from \cite{ahfock20} which are valid for large $n$. We use a simulation study to verify these results.

\section{Inference for Complete Sketching}\label{sec:complete}
\subsection{Distribution of the sketch estimator over  repeated sketches }

We begin by providing an alternative proof from \cite{ahfock20} for the complete sketch estimator which utilizes the marginal and conditional properties of the Wishart distribution. 
Foreshadowing, this methodology yields the pivotal quantities that will follow.  Firstly, consider the following Wishart random variable 
\begin{equation} \label{sketch wishart}
\left(
\begin{array}{cc}
 \mathbf{y}_{s}^{\mathsf{T}} \mathbf{y}_s    & \mathbf{y}_{s}^{\mathsf{T}} \mathbf{X}_s     \\
   \mathbf{X}_{s}^{\mathsf{T}} \mathbf{y}_s  &  \mathbf{X}_{s}^{\mathsf{T}} \mathbf{X}_s     
\end{array}
\right)
\sim 
W_{p+1} \left\{  k,   \frac{1}{k} \left(
\begin{array}{ccc}
  \mathbf{y}^{\mathsf{T}} \mathbf{y}  &   \mathbf{y}^{\mathsf{T}} \mathbf{X} &   \\
 \mathbf{X}^{\mathsf{T}}  \mathbf{y}    &  \mathbf{X}^{\mathsf{T}}  \mathbf{X}  &   
\end{array} \right)  \right\} .
\end{equation}
Using the properties of the Wishart distribution \cite[see][Theorem 3.3.9 on page 94]{Gupta2000}, we have that $ \mathbf{X}_{s}^{\mathsf{T}}  \mathbf{X}_s \sim W_{p}\left(k,  \mathbf{X}^{\mathsf{T}}   \mathbf{X} /k \right)$.
If we define $SSR_s = \mathbf{y}_{s}^{\mathsf{T}} \mathbf{y}_s -  \mathbf{y}_{s}^{\mathsf{T}} \mathbf{X}_s (\mathbf{X}_{s}^{\mathsf{T}} \mathbf{X}_s )^{-1} \mathbf{X}_s  \mathbf{y}_s $ then  $(\boldsymbol{\beta}_s, SSR_s)$ are independent. Furthermore,  
\begin{equation}
SSR_s
\sim W_{1} \left[ k-(p+1)+1,  SSR_F  /k \right],
\end{equation}
and using Theorem 4.2.1 from \cite{Gupta2000},  \begin{equation*}
 \boldsymbol{\beta}_s  \mid \mathbf{X}, \mathbf{y} 
\sim T_{p,1}\left\{ k-p+1,  \;  \boldsymbol{\beta}_F , \;
\left( \mathbf{X}^{\mathsf{T}} \mathbf{X} \right)^{-1}, SSR_F \right\},
\end{equation*}
where $T_{p,1}$ is the matrix-$t$ distribution which is equivalent to the multivariate-$t$ found by \cite{ahfock20} and given in  (\ref{complete sketch dist}).

\subsection{Distribution of the sketch estimator over  repeated samples }

To derive a distribution for the sketch estimator over repeated samples, we consider the stochastic representation of the complete sketch estimator (\ref{complete sketch dist}) which is
\begin{equation} \label{stochastic complete sketch dist}
\boldsymbol{\beta}_s 
=  \boldsymbol{\beta}_F +  \left\{ \frac{SSR_F /(k-p+1)}{U/(k-p+1)}  \right\}^{1/2} \mathbf{Z},
 \end{equation}
 where  $ \mathbf{Z} \sim N\{ \mathbf{0}_p, (\mathbf{X}^{\mathsf{T}} \mathbf{X})^{-1} \} $ and $U \sim \chi^2_{k-p+1}$.  When considering random samples we have that $\boldsymbol{\beta}_{F}$ is random as well. Specifically we have  
 \begin{equation*} 
\boldsymbol{\beta}_F\sim N\{ \boldsymbol{\beta}_0, \sigma^2 (\mathbf{X}^{\mathsf{T}} \mathbf{X} )^{-1} \}
\quad \mbox{and} \quad
SSR_F/\sigma^2 = V \sim \chi^2_{n-p}
 \end{equation*}
  then 
\begin{equation}  \label{beta stochastic  relationship}
\boldsymbol{\beta}_s 
= \boldsymbol{\beta}_0 + \sigma \left( 1 + \frac{ V}{U }   \right)^{1/2} \left(  \mathbf{X}^{\mathsf{T}} \mathbf{X}  \right)^{-1/2}   \mathbf{Z}
\end{equation}
where $  \mathbf{Z}  \sim N( 0_p, I_p )$, $V \sim \chi_{n-p}^2$ and $U \sim \chi^2_{k-p+1}$. The distribution of $U/(U+V)$ follows a  beta distribution with parameters $\alpha=k-p+1$ and $\boldsymbol{\beta}=n-p$. The using Theorem \ref{beta and normal} (in Appendix \ref{app:complete}), the density for the complete sketch estimator, denoted by $f_s\left( \mathbf{b} \right)$, is 
\begin{equation*}
\begin{split}
& \frac{  \Gamma\left( \frac{n}{2} \right)  \Gamma\left( \frac{n+k-p+1}{2} \right) }{  \left( 2 \pi \right)^{p/2} \Gamma\left (\frac{n-p}{2} \right)  } \left| \mathbf{X}^{\mathsf{T}} \mathbf{X} \right|^{1/2}  \\
& \times M\left[ \frac{k-p+1}{2}, \frac{( n+k-p+1)}{2}, - \frac{ (\mathbf{b}-\boldsymbol{\beta}_0)^{\mathsf{T}} \mathbf{X}^{\mathsf{T}} \mathbf{X} (\mathbf{b}-\boldsymbol{\beta}_0) }{2\sigma^2} \right]
\end{split}
\end{equation*}
where the Kummer $M(a,b,z)$ is the confluent hypergeometric functions of the first kind \cite[see][Chapter 13]{NIST}. 
For large $n-p$, the ratio $U/(U+V)$ can be written as $(n-p)/U$,  so in this case the density is approximately 
\begin{equation} \label{approx complete sample estimator}
\boldsymbol{\beta}_s  \stackrel{approx}{\sim} t_{p}\left\{ k-p+1,  \boldsymbol{\beta}_0,
 \sigma^2  \times \frac{ n-p}{ k-p+1 } \times \left( \mathbf{X}^{\mathsf{T}} \mathbf{X} \right)^{-1}   \right\}.
\end{equation}
Note that when  $k-p+1$ becomes large, $\boldsymbol{\beta}_s$ under repeated sampling can be approximated with a Gaussian distribution.

\subsection{The distribution of the error}

An important quantity for inference that has been largely overlooked in the sketching literature are the residuals. A natural estimate of the regression model error are the sketched residuals 
\begin{equation*}
SSR_s = \mathbf{y}_{s}^{\mathsf{T}}   \left( \mathbf{I} - \mathbf{X}_s \left(\mathbf{X}_{s}^{\mathsf{T}}  \mathbf{X}_s \right)^{-1} \mathbf{X}_{s}^{\mathsf{T}} \right) \mathbf{y}_s .
\end{equation*}
We note that the conditional distribution of the sketch estimator  is
\begin{equation} \label{c sketch conditional normal}
\boldsymbol{\beta}_s \mid \mathbf{X}_s, \mathbf{X}, \mathbf{y}
\sim N_{p}\left\{  \boldsymbol{\beta}_{F}, \;
 \frac{ SSR_F }{k}    \left( \mathbf{X}^{\mathsf{T}}  \mathbf{X} \right)^{-1}   \right\}, 
\end{equation}
and then conditional on  $SSR_F$ we have  
\begin{equation*} 
\frac{SSR_{s}}{SSR_F/k} \mid SSR_F  \sim \chi^2_{k-p},
\quad
\frac{SSR_F}{\sigma^2}  \sim \chi^2_{n-p}.
\end{equation*}
Based on these pivotal quantities the expected value is
\begin{equation*}  
E\left( SSR_s \right)
= E\left\{ E\left(  \frac{SSR_s}{SSR_F/k} \mid SSR_F\right)  \frac{SSR_F}{k} \right\} 
= \sigma^2 \frac{ (k-p) ( n- p)}{ k }
\end{equation*}
and an unbiased estimator for $\sigma^2$ is
\begin{equation}  \label{sigma s estimator}
 SSR_s \times \frac{k}{(n-p)(k-p)}.
\end{equation}
We can apply Theorem \ref{gamma conditional gamma} (see Appendix \ref{app:complete}) to obtain the density for $SSR_s/\sigma^2$ to be
\begin{equation} \label{uncondtional distribution SSR_s}
h(u) = \frac{  u^{(n+k-2p)/4-1   } K_{(k-n)/2}\left( \sqrt{u} \right) }{  \Gamma\left( \frac{k-p}{2} \right) \Gamma\left(\frac{n-p}{2}\right) 2^{(n-p)/2 (k-p)/2 -1} }  
\end{equation}
for $u>0$. For large sample size $n$, $SSR/(n-p)$ converges to $\sigma^2$, so then $SSR_s/\{ \sigma^2 (n-p)/k \}$ is approximately $\chi_{n-p}^2$.

\subsection{Inference on the parameter $\boldsymbol{\beta}_{F}$ under repeated sketching}

Based on the conditional distribution (\ref{c sketch conditional normal}),
the sketch estimate conditional on the covariate sketch, a pivotal quantity for inference on $\boldsymbol{\beta}_{F}$ is 
\begin{equation*}
\frac{ \left( \boldsymbol{\beta}_s -  \boldsymbol{\beta}_{F}\right)^{\mathsf{T}}  \left( \mathbf{X}_{s}^{\mathsf{T}} \mathbf{X}_s \right) \left( \boldsymbol{\beta}_s -  \boldsymbol{\beta}_{F}\right)}{SSR_F/k} 
\sim  \chi^2_p. 
\end{equation*}

Based on the  Wishart structure from (\ref{sketch wishart}) we have that $\boldsymbol{\beta}_s$ and $SSR_s$  are independent, 
which yields the following pivotal quantity 
\begin{equation*}
\frac{ \left( \boldsymbol{\beta}_s -  \boldsymbol{\beta}_{F}\right)^{\mathsf{T}}   \left( \mathbf{X}_{s}^{\mathsf{T}}  \mathbf{X}_s \right) \left( \boldsymbol{\beta}_s -  \boldsymbol{\beta}_{F}\right)/p}{SSR_s/(k-p)} 
\sim  F_{p, k-p}.
\end{equation*}
Then a marginal test for $j^{th}$ element $\boldsymbol{\beta}_{Fj}$ is based on 
\begin{equation} \label{marginal complete test}
\left\{ \boldsymbol{\beta}_{sj} -  \boldsymbol{\beta}_{Fj}\right) \times \left\{   \frac{ SSR_s}{k-p} \times  \left[ \left( \mathbf{X}_{s}^{\mathsf{T}}  \mathbf{X}_s \right)^{-1}\right]_{jj}   \right\}^{-1/2}  
\sim  t_{ k-p }
\end{equation}

\subsection{Inference on the parameter $\boldsymbol{\beta}_0$ under repeated samples}

Inference on $\boldsymbol{\beta}_0$ depends on what quantities one has access to. The simplest estimator can be obtained if we have $\mathbf{y}^{\mathsf{T}} \mathbf{y}$ and $\mathbf{W}_{\star} = \mathbf{S} \mathbf{S}^{\mathsf{T}} $. Here we assume that $\mathbf{W}_{\star}$ is non-singular. Consider the conditional distribution 
\begin{equation*}
\mathbf{y}_s   \mid  \mathbf{X}_s
\sim  N_p\left( \mathbf{X}_s \boldsymbol{\beta}_0, \sigma^2 \mathbf{S} \mathbf{S}^{\mathsf{T}} \right)
\sim  N_p\left( \mathbf{X}_s \boldsymbol{\beta}_0, \sigma^2  \mathbf{W}_{\star} \right) .
\end{equation*}
Then for inference on $\boldsymbol{\beta}_0$ we can use
\begin{equation*}
\frac{ \left( \boldsymbol{\beta}_s - \boldsymbol{\beta}_0 \right)^{\mathsf{T}} \mathbf{X}_{s}^{\mathsf{T}} \mathbf{W}_{\star}^{-1} \mathbf{X}_s   \left( \boldsymbol{\beta}_s - \boldsymbol{\beta}_0 \right)  }{ \sigma^2}  
\sim \chi^2_p. 
\end{equation*}
Then defining the idempotent matrix 
\begin{equation*}
\mathbf{H}_{\star} = \mathbf{S}^{\mathsf{T}} \mathbf{W}_{\star}^{-1} \mathbf{X}_s \left( \mathbf{X}_{s}^{\mathsf{T}} \mathbf{W}_{\star}^{-1} \mathbf{X}_s\right)^{-1}  \mathbf{X}_s \mathbf{W}_{\star}^{-1} \mathbf{S} 
\end{equation*}
and 
\begin{equation*}
SSR_{\star} 
=  \mathbf{y}^{\mathsf{T}}  \left( \mathbf{I}_n - \mathbf{H}_{\star}\right) \mathbf{y} 
= \mathbf{y}^{\mathsf{T}} \mathbf{y} -  \mathbf{y}_{s}^{\mathsf{T}}  \mathbf{W}_{\star}^{-1} \mathbf{X}_s \left( \mathbf{X}_{s}^{\mathsf{T}}  \mathbf{W}_{\star}^{-1} \mathbf{X}_s\right)^{-1}  \mathbf{X}_{s}^{\mathsf{T}}  \mathbf{W}_{\star}^{-1} \mathbf{y}_s . 
\end{equation*}
Since have $SSR_{\star}/\sigma^2  \sim \chi^2_{n-p}$, we can use the ratio to consider the following pivotal quantity 
\begin{equation*}
\frac{ \left( \boldsymbol{\beta}_s -  \boldsymbol{\beta}_0 \right)^{\mathsf{T}}   \left( \mathbf{X}_{s}^{\mathsf{T}}  \mathbf{W}_{\star}^{-1} \mathbf{X}_s \right) \left( \boldsymbol{\beta}_s -  \boldsymbol{\beta}_0 \right)/p}{SSR_{\star} /(n-p)} 
\sim  F_{p, n-p}.
\end{equation*}

If we do not have access to $\mathbf{W}_{\star}$, we can consider the  approximation $\mathbf{W}_{\star} \approx I_k \times n/k$ for large $n$.  Applying the approximation, the distribution of the sketch is
\begin{equation*}
\boldsymbol{\beta}_s  \mid \mathbf{X}_s  
\stackrel{approx}{\sim}
 N \left\{ \boldsymbol{\beta}_0 ,   \frac{   n  \sigma^2 }{k}\left( \mathbf{X}_{s}^{\mathsf{T}}  \mathbf{X}_s\right)^{-1} \right\}
\end{equation*}
and then approximate inference can be based on  
\begin{equation*} 
\frac{\left( \boldsymbol{\beta}_s -\boldsymbol{\beta}_0 \right)^{\mathsf{T}}  \left( \mathbf{X}_{s}^{\mathsf{T}} \mathbf{X}_s\right)  \left( \boldsymbol{\beta}_s -\boldsymbol{\beta}_0 \right)}{  n\sigma^2 /k
}
\stackrel{approx}{\sim}
\chi^2_p . 
\end{equation*}
To estimate $\sigma^2$ we can use (\ref{sigma s estimator}). To obtain a pivotal quantity, we require the distribution of the ratio of a $\chi_p^2$ and a random variable with density (\ref{uncondtional distribution SSR_s}). Applying Theorem \ref{gamma and beta divided by H} from Appendix \ref{app:complete}, we obtain the following distribution 
\begin{equation}
f(r) = \frac{  \left( \frac{n-p}{p} \right)^{p/2 - 1} \; \Gamma\left( \frac{ n}{2} \right)  \Gamma\left( \frac{k }{2} \right) }{  2^{(n-p)/2} \;  \Gamma\left( \frac{p}{2} \right)  \Gamma\left(\frac{k-p}{2}\right)  \Gamma\left( \frac{n-p}{2}\right)}   r^{ -p/2 } \;  U\left( \frac{n}{2},  \frac{n-k +2}{2}+1, \frac{n-p}{2k \; r} \right).
\end{equation}

In another approximation for large sample sizes $n$, we have $SSR/(n-p) \rightarrow \sigma^2$, so then $SSR_s/(\sigma^2 (k-p)/k ) \sim \chi_{k-p}^2$, 
\begin{equation*} 
\frac{\left( \boldsymbol{\beta}_s -\boldsymbol{\beta}_0 \right)  \left( \mathbf{X}_{s}^{\mathsf{T}}  \mathbf{X}_s\right)  \left( \boldsymbol{\beta}_s -\boldsymbol{\beta}_0 \right)/p}{ SSR_s /(k-p) }
\stackrel{approx}{\sim}
F_{p, k-p} 
\end{equation*}
and marginal test is based on 
\begin{equation} \label{marginal complete test samples}
\left\{ \boldsymbol{\beta}_{sj} -  \boldsymbol{\beta}_{j}\right) \times \left\{   \frac{ SSR_s}{k-p} \times  \left[ \left( \mathbf{X}_{s}^{\mathsf{T}}  \mathbf{X}_s \right)^{-1}\right]_{jj}   \right\}^{-1/2}  
\sim  t_{ k-p }.
\end{equation}

Finally,  another pivotal quantity can be obtained if one has access to $( \mathbf{X}^{\mathsf{T}} \mathbf{X} )$. Using the stochastic relationship in (\ref{beta stochastic  relationship}), we have the pivotal quantity is equivalent in distribution to
\begin{equation} \label{pivotal stochastic  relationship}
\frac{ \left( \boldsymbol{\beta}_s -  \boldsymbol{\beta}_0 \right)^{\mathsf{T}}  \left( \mathbf{X}^{\mathsf{T}} \mathbf{X} \right) \left( \boldsymbol{\beta}_s -  \boldsymbol{\beta}_0 \right)/p}{SSR_s \times \frac{k}{(n-p)(k-p)} } 
 =^d \frac{ V }{  U R   } \frac{(n-p)(k-p)}{k}
\end{equation}
where $V ~ \sim \chi^2_p$, $R \sim$ BETA$( k-p+1, n-p)$ and $U$ follows distribution in (\ref{uncondtional distribution SSR_s}).
Theorem \ref{gamma divided by H and U} in Appendix \ref{app:complete} gives the distribution of $V/(UR)$ as an integration involving the Kummer U function. Alternatively, one can  simulate the pivotal quantity in (\ref{pivotal stochastic  relationship}) and compare it to the observed value.

\subsection{Aside: A more efficient  complete sketch estimator}

If we have access to $\mathbf{W}_{\star}$ then we note that a more efficient estimator can be constructed via
\begin{equation*}
\boldsymbol{\beta}_s^{\star} = \left( \mathbf{X}_{s}^{\mathsf{T}}  \mathbf{W}_{\star}^{-1} \mathbf{X}_s\right)^{-1} \mathbf{X}_{s}^{\mathsf{T}}  \mathbf{W}_{\star}^{-1}  \mathbf{y}_s,
\end{equation*}
which has distribution
\begin{equation*}
\boldsymbol{\beta}_s^{\star}  \mid  \mathbf{S},   \mathbf{X}
\sim  N_p\left\{ \boldsymbol{\beta}_0, \sigma^2  \left( \mathbf{X}_{s}^{\mathsf{T}}   \mathbf{W}_{\star}^{-1} \mathbf{X}_s\right)^{-1} \right\}.
\end{equation*}
It can thus be shown  $\boldsymbol{\beta}_s^{\star}$ and $\boldsymbol{\beta}_s$ have the relation
\begin{equation*}
var \left( \boldsymbol{\beta}_s^{\star}  \right) 
\preceq var \left( \boldsymbol{\beta}_s  \right).
\end{equation*}
This relation can be obtained by considering the following  positive definite matrix
\begin{equation*}
\left[
\begin{array}{cc}
 \mathbf{A}  &  \mathbf{B}
  \\
  \mathbf{B}^{\mathsf{T}} & \mathbf{D}
\end{array}
\right]  
=
\left[
\begin{array}{cc}
 \left( \mathbf{X}_{s}^{\mathsf{T}}  \mathbf{W}_{\star}^{-1} \mathbf{X}_s\right)^{-1}    & \left( \mathbf{X}_{s}^{\mathsf{T}}  \mathbf{W}_{\star}^{-1} \mathbf{X}_s\right)^{-1} 
  \\
  \left( \mathbf{X}_{s}^{\mathsf{T}}  \mathbf{W}_{\star}^{-1} \mathbf{X}_s\right)^{-1}  &  \left( \mathbf{X}_{s}^{\mathsf{T}} \mathbf{X}_s\right)^{-1} \mathbf{X}_{s}^{\mathsf{T}} \mathbf{W}_{\star} \mathbf{X}_s  \left( \mathbf{X}_{s}^{\mathsf{T}} \mathbf{X}_s\right)^{-1}
\end{array}
\right] 
\succeq \mathbf{0},
\end{equation*}
then assuming that $\mathbf{A}$ is non-singular the Schur complement $\mathbf{D} - \mathbf{B}^{\mathsf{T}} \mathbf{A}^{-1} \mathbf{B} \succeq \mathbf{0}$. For large $n$ and $k$,  $\boldsymbol{\beta}_s$ and $\boldsymbol{\beta}_s^{\star}$ are asymptotically equivalent.

\section{Inference for Partial Sketching} \label{sec:part}

\subsection{ Distribution for the univariate partial sketch under repeated sketching}

Deriving distributions for partial sketch regression estimators and error is more challenging than in the complete sketch case because the support depends on the parameters. We begin with the univariate case as an illustration, then proceed with an approach that can provide approximate distributional forms.

When there is a single covariate with $ \mathbf{X}=\mathbf{x}$ and $\boldsymbol{\beta}_p$ is either inverse-gamma on $(0, \infty)$ or $(-\infty,0)$ depending on the sign of $\mathbf{X}^{\mathsf{T}} \mathbf{y}$. Continuing, suppose $\mathbf{X}^{\mathsf{T}} \mathbf{y} >0 $ then
\begin{equation*}
\boldsymbol{\beta}_p
 \sim inv\Gamma\left(  \frac{k}{2}, \boldsymbol{\beta}_F \frac{  \gamma k}{2}   \right)
\end{equation*}
where $inv\Gamma(a, b)$ denotes inverse gamma with shape $a$ and shape $b$. A pivotal quantity for inference in the univariate case ($p=1$) is 
\begin{equation} \label{partial with p equal to one}
\frac{ (k-p-1) \boldsymbol{\beta}_F   }{  \boldsymbol{\beta}_p  }  
 \sim \chi^2_k .
\end{equation}

\subsection{ The distribution of the error  }

One possible estimate for $SSR_F$ and $\sigma^2$ can be based on $\mathbf{y}^{\mathsf{T}} \mathbf{y} - SSM_p$ which would be unbiased but can be negative. An alternative estimate might to use the sum of squares of the  partial residuals, $ \mathbf{Y} -  \mathbf{X} \boldsymbol{\beta}_p$. It has the expectation 
\begin{equation} \label{sigma partial estimator}
E\left\{ \left(\mathbf{y} -  \mathbf{X} \boldsymbol{\beta}_p \right)^{\mathsf{T}}  \left(\mathbf{y} -  \mathbf{X} \boldsymbol{\beta}_p \right)  \right\}
= \mathbf{y}^{\mathsf{T}} \mathbf{y} +  SSM_F \left\{ \frac{ (k - p - 1)(p + 1) + 1}{(k - p)(k - p - 3)} - 1\right\}. 
\end{equation}
The derivation of this expectation is given in Appendix \ref{app2}. Finding an estimator that is unbiased for $SSR$ and $\sigma^2$ based on $(\mathbf{X}_{s}^{\mathsf{T}} \mathbf{X}_s)$ and $\mathbf{X}_{s}^{\mathsf{T}} \mathbf{y}$ that is also always positive is an open problem. Instead we suggest using the complete sketch estimate in (\ref{sigma s estimator}).

\subsection{The exact distribution for a linear combination under repeated sketches }

When $p > 2$ the support of the sketch estimator depends on the column space of $\mathbf{X}^{\mathsf{T}}  \mathbf{X}$. However, we can derive the exact distribution for a linear combination of the partial estimator under repeated sketching. This includes the marginal distribution for the $\boldsymbol{\beta}_{pj}$, the $j^{th}$ element of $\boldsymbol{\beta}_p$ which is the most common quantity of interest in practice for inferential purposes. 

The distribution is based on work by \cite{bodnar08}, in which a test for the weights of the global minimum variance portfolio in an elliptical model is considered. The test involves an inverse Wishart matrix and vector of ones. We adapt that result to consider an inverse Wishart matrix and any vector but we require the following assumption. 

\begin{assumption}[For test under repeated sketching]
\label{test repeated sketching assumption}
The vector $\mathbf{m}$ is not in the same direction as $\mathbf{X}^{\mathsf{T}} \mathbf{y}$. That is, we require 
\begin{equation*}
 \frac{ \mathbf{m}^{\mathsf{T}} (\mathbf{X}^{\mathsf{T}} \mathbf{y})}{  \lVert \mathbf{m}  \rVert \times \lVert \mathbf{X}^{\mathsf{T}} \mathbf{y}  \rVert } \neq  1.
\end{equation*}
\end{assumption}
Otherwise, we have that $\mathbf{m} = \mathbf{X}^{\mathsf{T}} \mathbf{y}$ and $\mathbf{m}^T \boldsymbol{\beta}_p$ follows an inverse gamma distribution and is equal to the estimate of the sum of squares from model. For partial sketch this estimate is
\begin{equation*}
 SSM_p = \mathbf{y}^{\mathsf{T}} \mathbf{X} \boldsymbol{\beta}_p  = \gamma \times \mathbf{y}^{\mathsf{T}} \mathbf{X} \left( \mathbf{X}_{s}^{\mathsf{T}} \mathbf{X}_s\right)^{-1} \mathbf{X}^{\mathsf{T}} \mathbf{y}.
\end{equation*}
Under Assumption \ref{test repeated sketching assumption}, a linear combination of the partial sketch estimator, $\mathbf{m}^{\mathsf{T}}  \boldsymbol{\beta}_p$, does not have closed from but has distribution equivalent to the random variables
\begin{equation} \label{partial beta estimator}
 \frac{1}{R} \left[   \mathbf{m}^{\mathsf{T}}  \boldsymbol{\beta}_F  + \left\{ \frac{ SSM_F \times  \mathbf{m}^{\mathsf{T}}  \left( \mathbf{X}^{\mathsf{T}} \mathbf{X} \right)^{-1} \mathbf{m} - \mathbf{m}^{\mathsf{T}} \boldsymbol{\beta}_F \boldsymbol{\beta}_F^{\mathsf{T}}  \mathbf{m} }{ (k-p+2) }  \right\}^{1/2} T \right],
\end{equation}
where $R \sim \Gamma( k-p+1, k-p-1)$  and $T$ follows a $t$-distribution with $k-p+1$ degrees of freedom and $R$ \& $T$ are independent. 
Then for the  $j^{th}$ element of $\boldsymbol{\beta}_p$ we have 
\begin{equation*}
\boldsymbol{\beta}_{pj}
=^d
 \frac{1}{R} \left(  \boldsymbol{\beta}_{Fj}  +  \left[ \frac{SSM_F \times  \left\{ \left(\mathbf{X}^{\mathsf{T}} \mathbf{X} \right)^{-1}  \right\}_{jj}  - \boldsymbol{\beta}_{Fj}^2    }{ (k-p+2) } \right]^{1/2} T \right).
\end{equation*}

For testing $\mathbf{m}^{\mathsf{T}}  \boldsymbol{\beta}_F = 0$, we can use the following pivotal quantity
\begin{equation*}
 \mathbf{m}^{\mathsf{T}}  \boldsymbol{\beta}_p \left\{ \frac{ (k-p+1) }{   SSM_p \gamma \times \mathbf{m}^{\mathsf{T}}  \left( \mathbf{X}_s \mathbf{X}_s \right)^{-1} \mathbf{m}  - (\mathbf{m}^{\mathsf{T}} \boldsymbol{\beta}_p)^2  } \right\}^{1/2}    
 \sim t_{k-p+1}
\end{equation*}
and then for testing a particular $\boldsymbol{\beta}_{Fj} =0$ we have the following pivotal quantity
\begin{equation} \label{partial beta equal to zero}
\boldsymbol{\beta}_{pj}  \left[ \frac{ (k-p+1)  }{   SSM_p \gamma \times  \left\{ \left(\mathbf{X}_s \mathbf{X}_s \right)^{-1}  \right\}_{jj}     - \boldsymbol{\beta}_{pj}^2  } \right]^{1/2}
 \sim t_{k-p+1}.
\end{equation}

\begin{itemize}
\item Remark: Although testing $\boldsymbol{\beta}_{Fj} =0$ might not strictly make traditional inferential sense, as it is testing if the sample estimate is equal to zero, however it may help to determine which variable coefficients are significantly different from zero. 
\end{itemize}

\subsection{The exact distribution for a linear combination under repeated sampling}

Again when $p > 2$ the support will depend on the column space of $\mathbf{X}^{\mathsf{T}}  \mathbf{X}$. We proceed similarly as in the previous section, but now consider the partial sketch estimate under repeated sampling and require a similar   assumption. 


\begin{assumption}[For test under repeated sampling]
The vector $\mathbf{m}$ is not in the same direction as $\boldsymbol{\beta}_0$. i.e. we require 
\begin{equation*}
 \frac{ \mathbf{m}^{\mathsf{T}} \boldsymbol{\beta}_0}{  \lVert \mathbf{m}  \rVert  \times \lVert \boldsymbol{\beta}_0  \rVert } \neq  1, 
\end{equation*}
\label{test repeated sampling assumption}
\end{assumption}
Under assumption \ref{test repeated sampling assumption}, a linear combination of the partial sketch estimator has distribution equivalent to the following random variables
\begin{equation*}
 \mathbf{m}^{\mathsf{T}}  \boldsymbol{\beta}_p
=^d
 \frac{k-p+1}{R} \left[   \mathbf{m}^{\mathsf{T}}  \boldsymbol{\beta}_0  +  \sigma  \left( 1  +\frac{ U }{V }  \right)^{1/2} \left\{ \mathbf{m}^{\mathsf{T}} \left( \mathbf{X}^{\mathsf{T}} \mathbf{X}\right)^{-1} \mathbf{m}  \right\}^{1/2} Z  \right].
\end{equation*}
where $R \sim \chi^2_{k-p+1}$, $Z \sim N(0,1)$,  $V \sim \chi^2_{k-p+2}$ and  $U$ follows a non-central $\chi^2_{p-1}$ with non-centrality parameter parameter equal to 
\begin{equation*}
\frac{ 1  }{ \sigma^2  } \left\{  \boldsymbol{\beta}_0^{\mathsf{T}}  \left( \mathbf{X}^{\mathsf{T}} \mathbf{X}\right) \boldsymbol{\beta}_0 - \frac{  ( \mathbf{m}^{\mathsf{T}}  \boldsymbol{\beta}_0)^2 }{ \mathbf{m}^{\mathsf{T}} \left( \mathbf{X}^{\mathsf{T}} \mathbf{X}\right)^{-1} \mathbf{m}  }  \right\}. 
\end{equation*}
This expression can alternatively be expressed using $U/(U+V)$ as this quantity is a non-central beta with parameters $p-1$ and $k-p+2$ with the same non-centrality parameter. 

For testing $\mathbf{m}^{\mathsf{T}} \boldsymbol{\beta}_0 = 0$ an exact distribution that requires $\sigma^2$ is 
\begin{equation*}
 \mathbf{m}^{\mathsf{T}}  \boldsymbol{\beta}_p \left\{ \frac{ (k-p+1) }{  \sigma^2  +   SSM_p \gamma \times \mathbf{m}^{\mathsf{T}}  \left( \mathbf{X}_s \mathbf{X}_s \right)^{-1} \mathbf{m}   - (\mathbf{m}^{\mathsf{T}} \boldsymbol{\beta}_p)^2  } \right\}^{1/2}    
 \sim t_{k-p+1}. 
\end{equation*}
To overcome this, we could estimate $\sigma^2$ using $SSR_s k/\{(n-p)(k-p) \}$ since it is independent of $( \mathbf{X}_s^{\mathsf{T}} \mathbf{X}_s )$. Then if $SSM_F$ or $n,k$ is large then an approximate test can be based on 
\begin{equation} \label{partial beta test samples}
 \mathbf{m}^{\mathsf{T}}  \boldsymbol{\beta}_p \left\{ \frac{ (k-p+1) }{  SSR_s \frac{k}{(n-p)(k-p)}  +   SSM_p \gamma \times \mathbf{m}^{\mathsf{T}}  \left( \mathbf{X}_s \mathbf{X}_s \right)^{-1} \mathbf{m}   - (\mathbf{m}^{\mathsf{T}} \boldsymbol{\beta}_p)^2  } \right\}^{1/2}    
\end{equation}
which follows a $t_{k-p+1}$.

\subsection{An approximate distribution for the partial sketch under repeated sampling}

Here we consider an approximation to the  partial sketch distribution under repeated sampling. 
When conditioning on $\mathbf{X}_s$, the partial sketch solution is 
\begin{equation*}
\boldsymbol{\beta}_p | \mathbf{X}_s  \sim N \left\{   \gamma  \left( \mathbf{X}_{s}^{\mathsf{T}}  \mathbf{X}_s \right)^{-1} \mathbf{X}^{\mathsf{T}}  \mathbf{X} \boldsymbol{\beta}_0, \gamma^2 \sigma^2 \times \left( \mathbf{X}_{s}^{\mathsf{T}}  \mathbf{X}_s \right)^{-1}  \mathbf{X}^{\mathsf{T}}  \mathbf{X} \left( \mathbf{X}_{s}^{\mathsf{T}}  \mathbf{X}_s \right)^{-1}  
\right\}. 
\end{equation*}
If for large $k$ we apply the approximation  $ \gamma \left( \mathbf{X}_s \mathbf{X}_s \right)^{-1}  \mathbf{X}^{\mathsf{T}} \mathbf{X} \left( \mathbf{X}_s \mathbf{X}_s \right)^{-1} \approx \left( \mathbf{X}_s \mathbf{X}_s \right)^{-1}$ then we can obtain an approximate density of the partial sketch. This approximate density, $f_p(\mathbf{b})$, for the partial is
\begin{equation} \label{approximate density unconditional}
\begin{split}
&\frac{  2^{ (p-k)/2}  \Gamma_p\left( \frac{k+1}{2} \right) k^{-p(k+1)/2}  }{   (\pi \gamma \sigma^2 )^{p/2} \Gamma_p\left(\frac {k}{2}\right )  \Gamma\left( \frac{k-p}{2} +1\right) } \\ 
& \times  \frac{ e^{ b^{\mathsf{T}}   \left( \mathbf{X}^{\mathsf{T}} \mathbf{X} \right) \boldsymbol{\beta}_0 /\sigma^2 }   K_{(p-k-2)/2}\left[ \left\{ \frac{\left( \mathbf{b}^{\mathsf{T}}  \mathbf{X}^{\mathsf{T}} \mathbf{X} \boldsymbol{\beta}_0 \right)^2}{\gamma \sigma^4 } + k \frac{ \boldsymbol{\beta}_0^{\mathsf{T}}    \mathbf{X}^{\mathsf{T}} \mathbf{X}  \boldsymbol{\beta}_0}{  \sigma^2 } \right\}^{1/2} \right]  }{ 
  \left( 1 + \frac{  \mathbf{b}^{\mathsf{T}}  \mathbf{X}^{\mathsf{T}} \mathbf{X} b }{k \gamma \sigma^2 }\right)^{(k+1)/2} 
\left\{ \frac{ \left(  \mathbf{b}^{\mathsf{T}}    \mathbf{X}^{\mathsf{T}} \mathbf{X} \boldsymbol{\beta}_0 \right)^2 }{\gamma \sigma^4} +k \frac{ \boldsymbol{\beta}_0^{\mathsf{T}}    \mathbf{X}^{\mathsf{T}} \mathbf{X} \boldsymbol{\beta}_0 }{ \sigma^2 } \right\}^{(p-k-2)/4}  }.
\end{split}
\end{equation}
The derivation is given in the Appendix~\ref{app2}. 
We note that the density in (\ref{approximate density unconditional}) is a special case of the matrix-variate generalized hyperbolic distribution.   \cite{thabane2004} introduce the matrix-variate generalized hyperbolic distribution by compounding the matrix normal distribution with the matrix generalized inverse Gaussian density of the scale matrix. They then consider Bayesian analysis of the matrix-variate generalized hyperbolic distribution in the normal multivariate linear model. However, this particular form of (\ref{approximate density unconditional}) has not been studied before to our knowledge.

Finally, an approximate density for the partial sketch under repeated sketching can also be constructed by applying some jitter or gaussian noise to either $ \mathbf{Y}$ or $\mathbf{X}^{\mathsf{T}} \mathbf{y}$ so that $ \mathbf{X}^{\mathsf{T}} \mathbf{y}_p \sim N( \mathbf{X}^{\mathsf{T}} \mathbf{y}, \delta \mathbf{X}^{\mathsf{T}} \mathbf{X} )$  or  $ (\mathbf{X}^{\mathsf{T}} \mathbf{y})_p \sim N( \mathbf{X}^{\mathsf{T}} \mathbf{y}, \delta \mathbf{I}_n )$ for some small $\delta$. 

\section{Extension to the Hadamard or Clarkson-Woodruff sketch}

The extensions are based on the following assumption and Theorem 2 \cite{ahfock20} which establishes a central limit for the Hadamard or Clarkson-Woodruff sketch.
For some positive definite matrix matrix, $\mathbf{Q}$, we have that
\begin{equation*} 
\lim_{n \rightarrow \infty}  \frac{1}{n}
\left[
\begin{array}{cc}
 \mathbf{y}^{\mathsf{T}} \mathbf{y}    & \mathbf{y}^{\mathsf{T}} \mathbf{X}     \\
   \mathbf{X}^{\mathsf{T}} \mathbf{y}  &  \mathbf{X}^{\mathsf{T}} \mathbf{X}    
\end{array}
\right]
= 
\left[
\begin{array}{cc}
 \mathbf{Q}_{yy}    & \mathbf{Q}_{yx}     \\
  \mathbf{Q}_{xy}   &  \mathbf{Q}_{xx}
\end{array}
\right]
\end{equation*}
Then based on this assumption, Theorem 2 from \cite{ahfock20} and the continuous mapping theorem we have the following convergence in distribution 
\begin{equation}
 \frac{1}{n} \left(
\begin{array}{cc}
 \mathbf{y}_{s}^{\mathsf{T}} \mathbf{y}_s    & \mathbf{y}_{s}^{\mathsf{T}} \mathbf{X}_s     \\
   \mathbf{X}_{s}^{\mathsf{T}} \mathbf{y}_s  &  \mathbf{X}_{s}^{\mathsf{T}} \mathbf{X}_s     
\end{array}
\right) \rightarrow W_{p+1} \left\{ k, \frac{1}{k}
\left[
\begin{array}{cc}
 \mathbf{Q}_{yy}    & \mathbf{Q}_{yx}     \\
  \mathbf{Q}_{xy}   &  \mathbf{Q}_{xx}
\end{array}
\right] 
\right\}
\end{equation}
Then applying the continuous mapping theorem we have convergence in distribution for the results in the previous two section while exchanging any sample quantity with $\mathcal{Q}$.  e.g. \cite{ahfock20} note in the supplementary material that $( n^{-1}  \mathbf{X}_{s}^{\mathsf{T}} \mathbf{X}_s )^{-1}$ has an inverse Wishart distribution with degrees of freedom equal to $k-p-1$ and matrix parameter equal to $k \mathbf{Q}_{xx}^{-1}$.

\section{Simulation Study}

Since the pivotal quantities are approximate when using the Hadamard or Clarkson-Woodruff sketch, we perform a simulation study to assess the approximation. We also include the Gaussian sketch as a way to gauge the approximation, as the distributions are exact in this case. We generate data from an artificial linear model so that we can consider repeated sketches and samples. Specifically, we simulate from a regression model with 
\begin{equation*}
\beta = \left(-5, -4,  -3, -2, -1, 0, 1, 2, 3, 4, 5\right) 
\quad \mbox{and}\quad 
\sigma^2=1.
\end{equation*}
We generate samples of size, $n$, equal to $10^4$ and sketches of size $k=21$ so that $k-p=10$. This could be considered to be a low number of sketches, but in general the sample size will depend on both the level of precision the user desires and their computational resources.  We shall focus on estimators for the parameters $\beta_1= -5$ and $\beta_6=0$. 

Generating a single sample with sample size ($n=10^4$) used for repeated sketching yields
\begin{equation*}
\beta_F = (-4.98, -3.98, -3.03, -1.99, -0.96,  0.03,  0.94,  2.01,  2.94,  4.02,  4.99)
\end{equation*}
and $\sigma^2_F =  0.99$. Next, we generate $m=10^4$ sketches (of size $k=21$) and consider the distribution of the complete and partial estimators for $\beta_1$ and $\beta_6$  in Figure~\ref{fig:beta1}. These histograms show agreement between theoretical and simulated distributions. The theoretical distributions of the complete and partial sketch estimators are given in (\ref{complete sketch dist}) and on (\ref{partial beta estimator}), respectively. As the empirical histograms appear to match the density curves in all cases, it seems that the approximate distributions for the Hadamard and Clarkson-Woodruff sketches are reasonable.

\begin{figure}[!htbp]
\includegraphics[width=5.5in, height=2in]{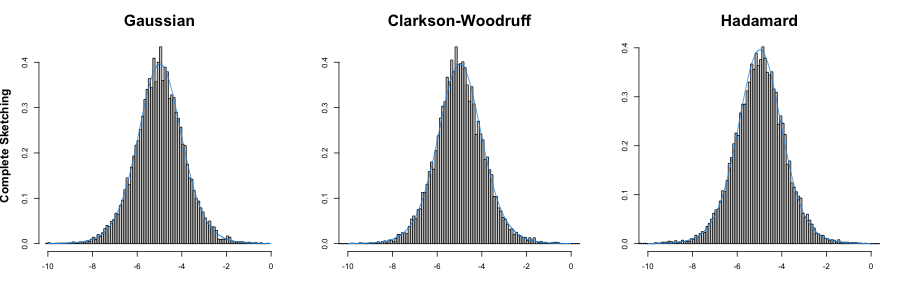}
\includegraphics[width=5.5in, height=2in]{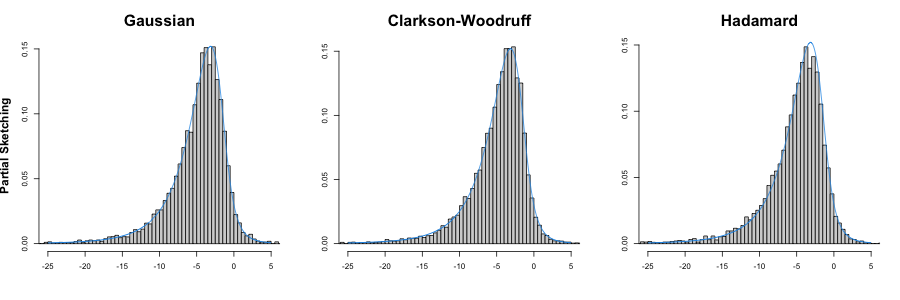}
\centering
\caption{Histograms under repeated sketching of the simulated complete and partial sketch estimator ($\tilde{\beta}_{s1}$  \& $\tilde{\beta}_{p1}$) using three different sketches under repeated sketching and an overlap of the density given in (\ref{partial beta estimator}). The Gaussian sketch is exact and the other two are approximate. The distribution of the complete and partial sketch estimators are based on (\ref{complete sketch dist}) and on (\ref{partial beta estimator}), respectively. }
\label{fig:beta1}
\end{figure}

In Figure \ref{fig:sketchPivotalbetaF}, we consider the distribution of the pivotal quantities for testing if $\beta_{F1}=0$ and $\beta_{F6}=0$ under both complete and partial sketching. The complete and partial sketch tests are based on (\ref{marginal complete test}) and (\ref{partial beta equal to zero}), respectively.  We see that for testing $\beta_{F6}=0$, both the complete and partial sketch approaches are performing approximately the same. However, for testing $\beta_{F1}=0$, it seems the partial sketch has substantially less power to detect the departure from the null hypothesis. This makes sense because, as evident in Figure~\ref{fig:beta1}, the partial estimator is a less efficient estimator for $\beta_{F1}$.  Figure \ref{fig:testingBeta6} in Appendix \ref{app3} show that approximate distribution for testing if $\beta_{F6}=0$ is reasonable when we change the type of sketch from the Gaussian to the Hadamard or Clarkson-Woodruff.

\begin{figure}[!htbp]
\includegraphics[width=5in, height=2.5in]{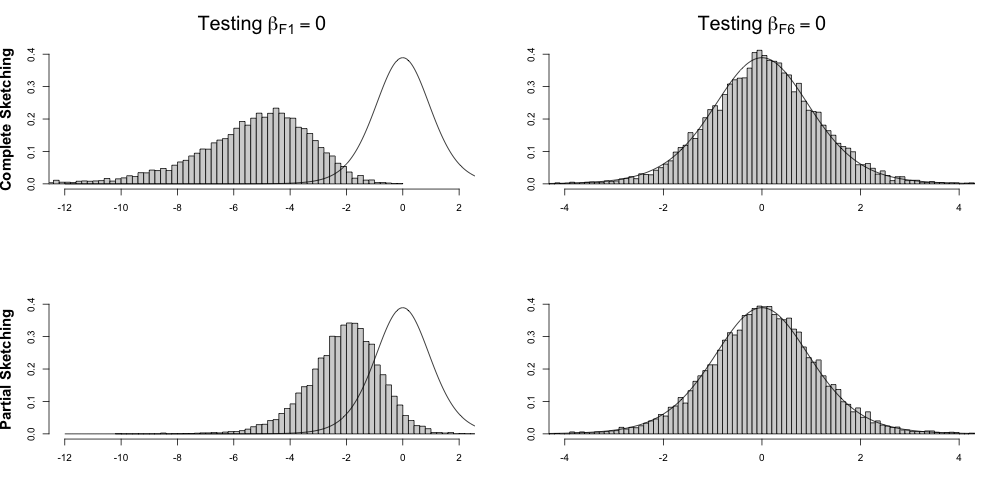}
\centering
\caption{Histograms under repeated sketching of the simulated pivotal quantity for testing if $\beta_{F1}=0$ (left column) and if $\beta_{F6}=0$ (right column) when performing complete sketching (top row) and when performing partial sketching (bottom row). The complete and partial sketch test are based on (\ref{marginal complete test}) and (\ref{partial beta equal to zero}), respectively. The density curve shows the distribution under the null hypothesis.  }
\label{fig:sketchPivotalbetaF}
\end{figure}

For repeated samples, we fix the covariate matrix and generate a new response vector and a sketching matrix for each replication. The approximate distributions are similar to the repeated sketching case.  Figure~\ref{fig:sketchPivotalbeta} shows the histograms of the simulated pivotal quantity for testing if $\beta_{1}=0$ and $\beta_{6}=0$ when we consider both complete and partial sketching. Figure \ref{fig:testingBeta6samples} in Appendix \ref{app3} show the approximate distribution for testing if $\beta_{6}=0$ is reasonable when we change the type of sketch from the Gaussian to the Hadamard or Clarkson-Woodruff sketch.

\begin{figure}[!htbp]
\includegraphics[width=5in, height=2.5in]{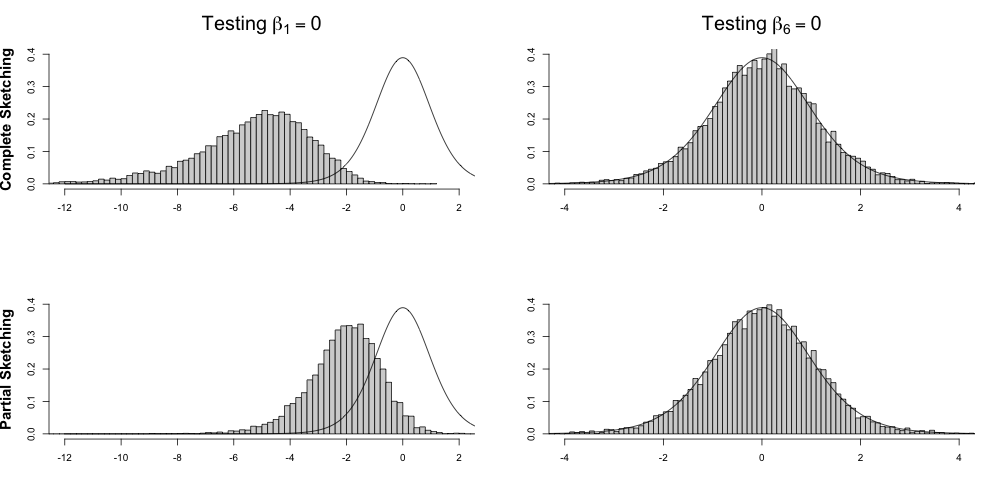}
\centering
\caption{Histograms under repeated samples of the simulated pivotal quantity for testing if $\beta_{1}=0$ (left column) and if $\beta_{6}=0$ (right column) when performing complete sketching (top row) and when performing partial sketching (bottom row). The complete and partial sketch test are based on (\ref{marginal complete test samples}) and (\ref{partial beta test samples}), respectively. The density curve shows the distribution under the null hypothesis.   }
\label{fig:sketchPivotalbeta}
\end{figure}

\section{Conclusion}

We derive the pivotal quantities for complete and partial estimators under repeated sketches and samples. We extended these results to include the Hadamard and Clarkson--Woodruff sketches then in a simulation study showed that the approximations hold well even when $k-p$ is small.



\newpage 
\appendix

\section{Complete Sketching Proofs} \label{app:complete}

\begin{theorem}
\label{beta and normal}
If $ \mathbf{Y} \sim N( \boldsymbol{\mu}, \boldsymbol{\Sigma}/R )$ and $R \sim BETA(\alpha, \eta)$  then the distribution of $ \mathbf{Y}$ is
\begin{equation} \label{den beta and normal}
h(\mathbf{y}) = \frac{  \Gamma(\alpha+\eta)  \Gamma( \eta +p/2) }{ \left( 2 \pi \right)^{p/2}   \left| \boldsymbol{\Sigma} \right|^{1/2}  \Gamma(\eta) }   M\left[ \alpha, \eta+p/2+\alpha, - (\mathbf{y} -\boldsymbol{\mu})^{\mathsf{T}} \boldsymbol{\Sigma}^{-1} (\mathbf{y} -\boldsymbol{\mu})/2 \right] 
\end{equation}
\end{theorem}

\begin{proof}[Proof of Theorem~\ref{beta and normal}]
Let $\mathbf{z}= \boldsymbol{\Sigma}^{-1/2} ( \mathbf{y} -\boldsymbol{\mu})$ then $\mathbf{z}\sim N( 0, \mathbf{I}_p/R )$. The marginal distribution of $\mathbf{z}$ is 
\begin{equation*}
\begin{split}
h(\mathbf{z}) &= \int_0^1 g(\mathbf{z}\mid r) b(r) du \\
&= \int_0^1 \left( 2 \pi \right)^{-p/2} r^{p/2}  \exp\left( -  r \frac{ \mathbf{z}^{\mathsf{T}} \mathbf{z} }{2}  \right)   \frac{r^{\alpha-1} (1-r)^{\eta-1}}{B(\alpha, \eta )}dr \\
&=  \frac{ \left( 2 \pi \right)^{-p/2} B(\alpha+p/2, \eta ) }{ B(\alpha, \eta) }  \int_0^1   \exp\left( - r  \frac{ \mathbf{z}^{\mathsf{T}} \mathbf{z}}{2}  \right)  \frac{ r^{\alpha+p/2-1} (1-r)^{\eta-1}}{ B(\alpha+p/2, \eta) } dr  \\
&=  \frac{ \left( 2 \pi \right)^{-p/2} B(\alpha+p/2,\eta) }{ B(\alpha,\eta) }  M_{BETA}\left( - r  \frac{ \mathbf{z}^{\mathsf{T}} \mathbf{z}}{2} , \alpha+p/2, \eta \right)  
 \end{split}
\end{equation*}
which involves the moment generating function of a beta distribution. The moment generating function can be expanded using the Maclaurin series of $\exp(x)= \sum_{m=0}^\infty \frac{x^m}{m!}$ to obtain
\begin{equation*}
h(\mathbf{z}) =   \left( 2 \pi \right)^{-p/2}  \sum_{m=0}^\infty \frac{1}{m!} \left( -\frac{ \mathbf{z}^{\mathsf{T}} \mathbf{z}}{2}  \right)^m    \frac{ B(\alpha+p/2+m,\eta)  }{ B(\alpha,\eta)  } .
\end{equation*}
Alternatively, we use the confluent hypergeometric functions of the first kind or the Kummer M,  ${}_{1}F_{1}(a,b,z) = M(a,b,z)$ because it has the following integral representation
 for $a>0$ and $b>0$
\begin{equation*}
M(a, b, z) = \frac{1}{\Gamma(a) \Gamma(b-a) } \int_0^{1} e^{ z t } t^{a-1} (1-t)^{b-a-1} dt .
\end{equation*}
This yield another expression for the marginal distribution of $z$.
\begin{equation*}
\begin{split}
h(\mathbf{z}) &=  \frac{ \left( 2 \pi \right)^{-p/2}  \Gamma(\alpha+\eta) }{ \Gamma(\alpha) \Gamma(\eta)  }  \int_0^1   \exp\left( - r \frac{ \mathbf{z}^{\mathsf{T}} \mathbf{z}}{2}  \right)   r^{\alpha-1} (1-r)^{ (\eta+p/2+\alpha) -\alpha  -1}  \\
  &=  \frac{ \left( 2 \pi \right)^{-p/2}  \Gamma(\alpha+\eta) }{ \Gamma(\alpha) \Gamma(\eta)  }   \Gamma(\alpha) \Gamma( \eta +p/2)   M(\alpha, \eta+p/2+\alpha, -\mathbf{z}^{\mathsf{T}} \mathbf{z}/2 ) .
 \end{split}
\end{equation*}
Then applying the change of variables $ \mathbf{Y}= \boldsymbol{\mu} + \boldsymbol{\Sigma}^{1/2} \mathbf{z}$ yields  the density in  (\ref{den beta and normal}).
\end{proof}

\begin{theorem}
\label{gamma conditional gamma}
If $U | V=v \sim \Gamma( k= \alpha, \theta=2v )$ and $V \sim \Gamma( k= \lambda, \theta=2)$ then the density of $U$ is 
\begin{equation} \label{density gamma conditional gamma} 
H\left( u \mid \lambda, \alpha \right) 
 = \frac{ 2^{-(\lambda+ \alpha) +1} }{  \Gamma(\alpha) \Gamma(\lambda) }  u^{(\lambda+\alpha)/2-1} K_{\alpha-\lambda}\left( \sqrt{u}   \right) 
\end{equation}
for $u >0$. 
\end{theorem}

\begin{proof}[Proof of Theorem~\ref{gamma conditional gamma}]
The density of $U$ when $U \sim \Gamma( k= \alpha, \theta=2y )$ and $V \sim \Gamma( k= \lambda, \theta=2)$ is 
\begin{equation*}  \begin{split}
f(x) 
&= \int_{0}^{\infty} g(u \mid  \nu, 2v ) g(v \mid \psi, 2 ) dy \\ 
&= \int_{0}^{\infty}  \frac{1}{ (2v)^{\alpha}  \Gamma(\alpha)} u^{\alpha-1}e^{-u/(2v) }  \frac{1}{ 2^{\lambda}  \Gamma(\lambda)} v^{\lambda-1}e^{-v/2 } dv \\ 
& =  \frac{ 2^{-(\lambda+ \alpha)} }{  \Gamma(\alpha) \Gamma(\lambda) }  u^{\alpha-1} \int_{0}^{\infty} v^{\lambda-\alpha-1} e^{- \frac{1}{2}( u/v +v) } dv \\ 
& = \frac{ 2^{-(\lambda+ \alpha) +1} }{  \Gamma(\alpha) \Gamma(\lambda) }  u^{(\lambda+\alpha)/2-1} K_{\alpha-\lambda}\left( \sqrt{u} \right)  \\ 
\end{split}
\end{equation*}
where $K_{\nu}(z)$ is the modified Bessel function of the second kind, \cite[See Chapter 10]{NIST}. 
\end{proof}

\begin{remark}
The density (\ref{density gamma conditional gamma}) is related to integral for the $K_{\nu}(z)$ \cite[Eq.~10.43.19]{NIST}. 
The moments of this density are
\begin{equation*}
E \left[U^k \right] = \frac{ 2^{2k}  \Gamma(k + \alpha)\Gamma(\lambda + k)}{\Gamma(\alpha)\Gamma(\lambda) } .
\end{equation*}
\end{remark}

\begin{theorem}
\label{gamma and beta divided by H}
If  $Q \sim \Gamma( k= \phi, \theta=2 )$ and $U \sim H(\alpha, \lambda)$ density from (\ref{density gamma conditional gamma}) then the distribution of the ratio, $Q/U$, is 
\begin{equation} \label{pivotal gamma divided H}
f(r) = \frac{  \Gamma( \alpha +\phi ) \Gamma( \lambda + \phi)   }{ 2^{\lambda} r^{ \phi } \Gamma(\phi)  \Gamma(\alpha)  \Gamma(\lambda)}  U\left( \lambda + \phi,  \lambda -\alpha +1, \frac{1}{2r} \right) 
\end{equation}
where $U(a,b,z)$ is the Kummer U-function \cite[See Eq. 13.4.4]{NIST}. 
\end{theorem}

\begin{proof}[Proof of Theorem~\ref{gamma and beta divided by H}]
The density of the ratio integral and $H(\alpha, \lambda)$ can found by integrating out the following
\begin{equation*}  \begin{split}
f(r)  & =  \int_0^\infty u f_Q( r u) f_U(u) du\\
 & =  \int_0^\infty u f_Q( ru )  \int_0^\infty f_{U|V}( u ) f_{V}( v) du du\\
 & =   \int_0^\infty u  \frac{(ru)^{\phi-1}e^{-ru/2 }}{ 2^{\phi}  \Gamma(\phi)}   \int_{0}^{\infty}  \frac{u^{\alpha-1}e^{-u/(2v) }}{ (2v)^{\alpha}  \Gamma(\alpha)}   \frac{ v^{\lambda-1}e^{-v/2 }}{ 2^{\lambda}  \Gamma(\lambda)} dv du \\
  & =   \frac{2^{-(\lambda+\alpha+\phi)} r^{\lambda-1} }{  \Gamma(\phi)  \Gamma(\alpha)  \Gamma(\lambda)} \int_0^\infty   v^{\lambda -\alpha-1}e^{-v/2 }  \int_{0}^{\infty}  u^{\alpha+\phi-1}e^{-u(1/v +r)/2  }   du dv \\
   & =   \frac{2^{-(\lambda+\alpha+\phi)} r^{\lambda-1} }{  \Gamma(\phi)  \Gamma(\alpha)  \Gamma(\lambda)} \int_0^\infty   v^{\lambda -\alpha-1}e^{-v/2 }  \Gamma( \alpha +\phi ) 2^{\alpha +\phi}\left( 1/v +r \right)^{ -(\alpha +\phi) } dv \\
     & =   \frac{2^{-\lambda} r^{\lambda-1}  \Gamma( \alpha +\phi )  }{  \Gamma(\phi)  \Gamma(\alpha)  \Gamma(\lambda)} \int_0^\infty  e^{-v/2 }  v^{\lambda + \phi -1}  \left( 1 +r v \right)^{ -(\alpha +\phi) } dv. 
\end{split} \end{equation*}
Now performing the change of variables $ \mathbf{Y} = r v$, $v=y/r$ and then $dv = 1/r dy$ we obtain
\begin{equation*}
 \frac{2^{-\lambda} r^{- \phi }  \Gamma( \alpha +\phi )  }{  \Gamma(\phi)  \Gamma(\alpha)  \Gamma(\lambda)} \int_0^\infty  e^{- y/(2r) }  y^{\lambda + \phi -1} \left( 1 + \mathbf{y} \right)^{ -(\alpha +\phi) } dy.
\end{equation*}
This expression is related to the Kummer U function \cite[See Eq. 13.4.4]{NIST} which the following integral representation
 for $a>0$,
\begin{equation*} 
U(a, b, z) = \frac{1}{\Gamma(a) } \int_0^{\infty} e^{ - z t } t^{a-1} (1+t)^{b-a-1} dt.
\end{equation*}
Applying the Kummer U function integral representation we obtain (\ref{pivotal gamma divided H})
\end{proof}

\begin{theorem}
\label{gamma divided by H and U}
If  $ X \sim \Gamma( k= \phi, \theta=2 )$, $V \sim BETA(\kappa,\beta)$ and $U \sim H(\alpha, \lambda)$ follows the density from (\ref{density gamma conditional gamma}) then the distribution of  $R=Q/(VU)$ is 
\begin{equation*}  
f(r)   = C(\kappa, \beta, \alpha, \phi )  r^{-\phi}  \int_0^\infty U\left( \lambda + \phi,  \lambda -\alpha +1, \frac{1}{2ru} \right)   u^{\kappa-\phi} (1-u)^{\boldsymbol{\beta}-1}  du
\end{equation*}
where 
\begin{equation*} 
C(\kappa, \beta, \alpha, \phi) = \frac{\Gamma(\kappa+\beta)}{\Gamma(\kappa)\Gamma(\beta)}     \frac{2^{-\lambda}   \Gamma( \alpha +\phi )  }{  \Gamma(\phi)  \Gamma(\alpha)  \Gamma(\lambda)} \Gamma( \lambda + \phi)
 \end{equation*}
\end{theorem}

\begin{proof}[Proof of Theorem~\ref{gamma divided by H and U}]
Let $Y = X/U$ then density of $Y$ is in (\ref{pivotal gamma divided H}). Then the density of the ratio $Y/Z$ is obtained similar to the proof of theorem \ref{gamma and beta divided by H}.
\end{proof}

\section{Partial Sketching Proofs} \label{app2}

\begin{proof}[Proof of equation (\ref{sigma partial estimator})]
Begin by expanding the quadratic form into three pieces, 
\begin{equation*}\begin{split}
& \left(\mathbf{y} -  \mathbf{X} \boldsymbol{\beta}_p \right)^{\mathsf{T}}  \left(\mathbf{y} -  \mathbf{X} \boldsymbol{\beta}_p \right)  \\
=\; & \mathbf{y}^{\mathsf{T}} \mathbf{y} - 2  \gamma \times \mathbf{y}^{\mathsf{T}}  \mathbf{X} (\mathbf{X}_{s}^{\mathsf{T}}  \mathbf{X}_s)^{-1} \mathbf{X}^{\mathsf{T}} \mathbf{y} + \gamma^2 \mathbf{y}^{\mathsf{T}} \mathbf{X}  (\mathbf{X}_{s}^{\mathsf{T}}  \mathbf{X}_s)^{-1} \mathbf{X}^{\mathsf{T}} \mathbf{X} (\mathbf{X}_{s}^{\mathsf{T}}  \mathbf{X}_s)^{-1} \mathbf{X}^{\mathsf{T}} \mathbf{y}. \\
\end{split}\end{equation*}
The first term is fixed and the expected value the middle is term is
\begin{equation*}  
E[\mathbf{y}^{\mathsf{T}}  \mathbf{X} \boldsymbol{\beta}_p  ] =  \mathbf{y}^{\mathsf{T}} \mathbf{X} \boldsymbol{\beta}_F = \mathbf{y}^{\mathsf{T}} \mathbf{X} \left( \mathbf{X}^{\mathsf{T}}  \mathbf{X} \right)^{-1} \mathbf{X}^{\mathsf{T}} \mathbf{y} = SSM_F.
\end{equation*}
The expected of the last term can be obtained using the $vec$ operator \begin{equation*}
\gamma^2 \left\{ vec  (\mathbf{X}_{s}^{\mathsf{T}}  \mathbf{X}_s)^{-1} \right\}^{\mathsf{T}}  \left( \mathbf{X}^{\mathsf{T}} \mathbf{y} \mathbf{y}^{\mathsf{T}} \mathbf{X}   \otimes  \mathbf{X}^{\mathsf{T}} \mathbf{X} \right)   vec (\mathbf{X}_{s}^{\mathsf{T}} \mathbf{X}_s)^{-1}.
\end{equation*}
Then expected value can be written as
\begin{equation*} 
\begin{split}
&\gamma^2 tr \left[ var\left\{ vec (\mathbf{X}_{s}^{\mathsf{T}} \mathbf{X}_s)^{-1} \right\}  \left( \mathbf{X}^{\mathsf{T}} \mathbf{y} \mathbf{y}^{\mathsf{T}} \mathbf{X}   \otimes  \mathbf{X}^{\mathsf{T}} \mathbf{X} \right)  \right] \\
&+  \gamma^2  \mathbf{y}^{\mathsf{T}} \mathbf{X}  E\left\{ (\mathbf{X}_{s}^{\mathsf{T}} \mathbf{X}_s)^{-1} \right\} \mathbf{X}^{\mathsf{T}} \mathbf{X} E\left\{ (\mathbf{X}_{s}^{\mathsf{T}} \mathbf{X}_s)^{-1} \right\} \mathbf{X}^{\mathsf{T}} \mathbf{y}
\end{split}
\end{equation*}
and the last term is 
\begin{equation*}
\gamma^2 \mathbf{y}^{\mathsf{T}} E\left\{ (\mathbf{X}_{s}^{\mathsf{T}} \mathbf{X}_s)^{-1} \right\} \mathbf{X}^{\mathsf{T}} \mathbf{X} E\left\{ (\mathbf{X}_{s}^{\mathsf{T}} \mathbf{X}_s)^{-1} \right\} \mathbf{X}^{\mathsf{T}} \mathbf{y} = SSM_F.
\end{equation*}
For the expression involves the variance of the inverse wishart, if $ (\mathbf{X}_{s}^{\mathsf{T}} \mathbf{X}_s)^{-1} \sim W_p^{-1} \{ n-p-1, k (\mathbf{X}^{\mathsf{T}} \mathbf{X})^{-1} \} $ then $var \{ vec (\mathbf{X}^{\mathsf{T}} \mathbf{X})^{-1} \} $ is 
\begin{equation*}
\begin{split}
\zeta(k,n,p) \times & \left\{ vec (\mathbf{X}^{\mathsf{T}} \mathbf{X})^{-1} \left[ vec (\mathbf{X}^{\mathsf{T}} \mathbf{X})^{-1} \right]^{\mathsf{T}} \right. \\
& \left. + (n-p-1) \left( I_{n^2} +  K_{nn}\right) \left[  (\mathbf{X}^{\mathsf{T}} \mathbf{X})^{-1} \otimes  (\mathbf{X}^{\mathsf{T}} \mathbf{X})^{-1} \right] \right\}
\end{split}
\end{equation*}
where
\begin{equation*}
\zeta(k,n,p) = \frac{k^2}{ \left( n -p \right) \left( n - p -1 \right)^2 \left( n -p -3 \right)}
\end{equation*}
and $K_{nn}$ is a commutation matrix. The commutation matrix, $K_{mn}$, is a unique $mn \times mn$ permutation matrix which gives an relation between the vec of a matrix and its transposed. i.e. if   $A$ is $m \times n$ matrix then  $K_{mn} vec A =vec A^{\mathsf{T}} $. See \cite{MagnusJanR1988Mdcw} for further properties on the  commutation matrix.  More some properties for inverse Wishart distribution are described within \cite{Gupta2000}. 

Then the first term of the numerator of the variance becomes
 \begin{equation*}
\gamma^2 k^2 tr\left[ vec (\mathbf{X}^{\mathsf{T}} \mathbf{X})^{-1} \left\{ vec (\mathbf{X}^{\mathsf{T}} \mathbf{X})^{-1} \right\}^{\mathsf{T}} (\mathbf{X}^{\mathsf{T}} \mathbf{y} \mathbf{y}^{\mathsf{T}} \mathbf{X})   \otimes  (\mathbf{X}^{\mathsf{T}}  \mathbf{X}) \right] 
=\gamma^2 k^2 SSM_F.
\end{equation*}
Then the second term of the numerator of the variance becomes
 \begin{equation*}
 \begin{split}
&\gamma^2 k^2 tr\left[  (n-p-1) \left( I_{n^2} +  K_{nn}\right) \left( (\mathbf{X}^{\mathsf{T}}  \mathbf{X})^{-1}  \otimes (\mathbf{X}^{\mathsf{T}}  \mathbf{X} )^{-1}   \right) \right]  \\
& =\gamma^2 k^2 (n-p-1) (p+1) SSM_F.
\end{split}
\end{equation*}
Combining these results gives the desired expression.
\end{proof}

\begin{proof}[Proof of approximate density (\ref{approximate density unconditional})]
We begin with letting $\eta = \gamma \sigma^2$
and
\begin{equation*}
\mathbf{X}_{s}^{\mathsf{T}} \mathbf{X}_s = \left( \mathbf{X}^{\mathsf{T}} \mathbf{X} \right)^{1/2}  \mathbf{B}  \left( \mathbf{X}^{\mathsf{T}} \mathbf{X} \right)^{1/2} 
\end{equation*}
where $\mathbf{B}  \sim W_p(  \frac{1}{k} I_p, k )$. Then we have that
 \begin{equation*}
 \begin{split}
\boldsymbol{\beta}_p | \mathbf{B}  \sim & \mathcal{N} \left( \gamma \left( \mathbf{X}^{\mathsf{T}} \mathbf{X} \right)^{-1/2}  \mathbf{B} ^{-1}   \left( \mathbf{X}^{\mathsf{T}} \mathbf{X} \right)^{+1/2} \boldsymbol{\beta}_0,  \right. \\
& \left. \eta \gamma \times \left( \mathbf{X}^{\mathsf{T}} \mathbf{X} \right)^{-1/2}  \mathbf{B} ^{-1}     \mathbf{B} ^{-1}   \left( \mathbf{X}^{\mathsf{T}} \mathbf{X} \right)^{-1/2} \right).
\end{split}
\end{equation*}
Then we let $\boldsymbol{\lambda} = \gamma  \left( \mathbf{X}^{\mathsf{T}} \mathbf{X} \right)^{+1/2} \boldsymbol{\beta}_0$ and perform the transformation of variables  $\mathbf{z}= \left( \mathbf{X}^{\mathsf{T}} \mathbf{X} \right)^{1/2} \boldsymbol{\beta}_p$ to obtain $\mathbf{z} | \mathbf{B}  \sim N ( \mathbf{B} ^{-1} \boldsymbol{\lambda}, \eta \mathbf{B} ^{-1}\mathbf{B} ^{-1} )$. However, integrating out $B$ does not yield a closed form density instead we approximate the stochastic relationship by noting that $E[\gamma \mathbf{B} ^{-1}] = I $ and apply the following approximation $\gamma \mathbf{B} ^{-1} \mathbf{B} ^{-1}  = \mathbf{B} ^{-1}$. This is not an optimal approximation but allows for a closed form density to be found and it performs well under simulations. 

Applying the approximation, we now tasked with deriving the marginal density for 
$ \mathbf{z} | \mathbf{B}  \sim N ( \mathbf{B} ^{-1} \boldsymbol{\lambda}, \eta \mathbf{B} ^{-1} )$
which has density 
\begin{equation*}
f(  \mathbf{z} | \mathbf{B} ) = \frac{ |  \mathbf{B}  |^{1/2}  }{ (2\pi  \eta)^{p/2}  }  \exp \left\{-\frac{1}{2 \eta } tr \left(  \mathbf{B}      \mathbf{z} \mathbf{z}^{\mathsf{T}}  -  2  \boldsymbol{\lambda} z^{\mathsf{T}}   +  \mathbf{B} ^{-1} \boldsymbol{\lambda} \boldsymbol{\lambda}^{\mathsf{T}}    \right) \right\}
\end{equation*}
and the density for $B$ is
 \begin{equation*}
f ( \mathbf{B}  ) 
= \frac{k^{pk/2}    \left|\mathbf{B} \right|^{(k-p-1)/2   }}{2^{kp/2}  \Gamma_p\left(\frac {k}{2}\right ) }  \exp \left\{ - \frac{k}{2} tr \left( \mathbf{B}  \right) \right\}.
\end{equation*}
The joint density is 
 \begin{equation*}
 f( \mathbf{z}, \mathbf{B} ) = \frac{ k^{pk/2} | \mathbf{B}  |^{(k-p)/2 }  }{ (2\pi)^{p/2} 2^{kp/2}  \Gamma_p\left(\frac {k}{2}\right )  }  \exp\left\{-\frac{1}{2 } tr \left(  \mathbf{B}      \mathbf{z} \mathbf{z}^{\mathsf{T}}  -  2 \boldsymbol{\lambda} z^{\mathsf{T}}   +B^{-1} \boldsymbol{\lambda} \boldsymbol{\lambda}^{\mathsf{T}}  + k  \mathbf{B}   \right)  \right\} .
\end{equation*}
We then integrate out $B$ from the joint density
\begin{equation*}
\begin{split}
& e^{ \mathbf{z}^{\mathsf{T}}   \boldsymbol{\lambda}/\eta} \frac{(2\pi \eta)^{-p/2}   }{ 2^{kp/2} \Gamma_p\left(\frac {n}{2}\right ) } \\
& \times \int 
 |  \mathbf{B} |^{(k+1)/2-(p+1)/2 } \exp\left[-\frac{1}{2} tr \left\{  \mathbf{B}  \left(   \mathbf{z}  \mathbf{z}^{\mathsf{T}} /\eta + k I \right)+ B^{-1} \boldsymbol{\lambda} \boldsymbol{\lambda}^{\mathsf{T}}  /\eta \right\} \right]
 d W  
\end{split}
\end{equation*}
the resulting integrate is related to the Bessel function of matrix argument. 
We write it in standard form by letting $\mathbf{A} =   \mathbf{z}  \mathbf{z}^{\mathsf{T}} /\eta + k \mathbf{I} $ and $a =\boldsymbol{\lambda} \eta^{-1/2}$ then we have
\begin{equation*} 
e^{ \mathbf{z}^{\mathsf{T}} \boldsymbol{\lambda}/\eta } \frac{(2\pi \eta)^{-p/2}  }{ 2^{kp/2}  \Gamma_p\left(\frac {k}{2}\right ) } \int 
 |  \mathbf{W}  |^{ (k+1)/2 -(p+1)/2 } \exp\left\{-\frac{1}{2} tr \left(    \mathbf{W}  \mathbf{A}+  \mathbf{W}^{-1} \mathbf{a}  \mathbf{a}^{\mathsf{T}}  \right) \right\}
 d W .
\end{equation*}
Applying the Bessel function of matrix argument yields the following
\begin{equation*}
e^{ \mathbf{z}^{\mathsf{T}} \boldsymbol{\lambda}/\eta } \frac{(2\pi \eta)^{-p/2}  }{ 2^{kp/2}  \Gamma_p\left(\frac {k}{2}\right ) } 
B_{-(k+1)/2}^{(p)} \left( \frac{1}{4} \mathbf{A} \mathbf{a} \mathbf{a}^{\mathsf{T}}   \right)  \left| \frac{1}{2} \mathbf{A} \right|^{-(k+1)/2}
\end{equation*}
where for the $B_{\nu}^{(p)}$ matrix Bessel function for dimension $p$,   \citep[See][p. 506]{herz1955}. Then from \citep[][p. 509]{herz1955}, if $V$ is a $p\times p$ matrix with rank $m$ then with $-\delta > (p-m-1)/2$, 
 \begin{equation*}
B_{-\delta}^{(p)}( \mathbf{V}  ) =  \frac{\Gamma_p(\delta) }{ \Gamma_m \left(\delta - \frac{p-1}{2} \right)} B_{-\delta + (p-m)/2 }^{(m)} (  \mathbf{V}   ) 
\end{equation*}
where $B_{-\nu + (p-m)/2 }^{(m)}$ depends on the $m$ eigenvalues of $V$. Using this connection formula we can obtain the following relation 
\begin{equation*} 
B_{-(k+1)/2}^{(p)}(  \mathbf{A}   \mathbf{a}   \mathbf{a}^{\mathsf{T}}   ) =  \frac{\Gamma_p \left( \frac{k+1}{2}\right) }{ \Gamma\left ( \frac{k+1}{2} - \frac{p-1}{2} \right)} B_{-(k+1)/2 + (p-1)/2 }^{(1)} (  \mathbf{a}^{\mathsf{T}}  \mathbf{A}  \mathbf{a}  ) 
\end{equation*}
and then using the relation which relation the matrix Bessel function to the Modified Bessel function of the second  
 \begin{equation*}
B_\nu^{(1)}  \left( \frac{y}{4} \right)= y^{-\nu/2} \; 2^{\nu+1} K_\nu( \sqrt{y} ). 
 \end{equation*}
Combing the last two equations we get the following connection formula
 \begin{equation*} 
 \begin{split}
&B_{-(k+1)/2}^{(p)} \left( \frac{1}{4}  \mathbf{A}   \mathbf{a}   \mathbf{a}^{\mathsf{T}}   \right) \\
&= \frac{\Gamma_p\left( \frac{k+1}{2} \right) 2^{ -(k+1)/2 + (p-1)/2 +1}  }{ \Gamma\left( \frac{k-p}{2} +1 \right)} ( \mathbf{a}^{\mathsf{T}}  \mathbf{A}  \mathbf{a})^{(k- p+2)/4} \; K_{-(k-p)/2}\left\{  \left(\mathbf{a}^{\mathsf{T}} \mathbf{A} \mathbf{a}\right)^{1/2} \right\}. 
\end{split}
 \end{equation*}
Applying this connection formula, the density now becomes
\begin{equation}  \label{almost b density}
 \begin{split}
&e^{ \mathbf{z}^{\mathsf{T}} \boldsymbol{\lambda}/\sigma } \frac{  2^{ -(k-p)/2}   }{ \Gamma_p\left(\frac {k}{2}\right ) (\pi \eta)^{p/2} } 
 \left|  \mathbf{A} \right|^{-(k+1)/2}  \\ 
& \times
\frac{\Gamma_p\left( \frac{k+1}{2} \right) }{ \Gamma\left( \frac{k-p}{2} +1  \right)} (\mathbf{a}^{\mathsf{T}} \mathbf{A} \mathbf{a})^{-(k+1)/4} \;  K_{(p-k)/2-1}\left\{  \left(\mathbf{a}^{\mathsf{T}} \mathbf{A} \mathbf{a}\right)^{1/2} \right\} .
\end{split}
\end{equation}
We now substitute $\mathbf{A}$ and $\mathbf{a}^{\mathsf{T}} \mathbf{A}\mathbf{a}$. For the term $ \left|   \mathbf{A} \right|$ we apply the matrix determinant lemma 
\begin{equation*}
 \left|  \mathbf{A} \right|
 =   \left|  \mathbf{z} \mathbf{a}^{\mathsf{T}} /\eta + k \mathbf{I} \right|
 =  \left\{ 1 +  \mathbf{z}^{\mathsf{T}} \mathbf{z} /(k\eta) \right\}  \left|   k  \mathbf{I} \right|
 =  \left\{ 1 + \mathbf{z}^{\mathsf{T}} \mathbf{z} /(k\eta) \right\}  k^{p}.
\end{equation*}
The term $a^{\mathsf{T}} \mathbf{A} \mathbf{a}$ is
\begin{equation*} 
a^{\mathsf{T}} \mathbf{A} \mathbf{a}
=  \left( \mathbf{z}^{\mathsf{T}}  \boldsymbol{\lambda}  \right)^2 /\eta^2+ k \boldsymbol{\lambda}^{\mathsf{T}}  \boldsymbol{\lambda}/\eta
\end{equation*}
To simplify the multivariate gamma we have the relation 
\begin{equation*} 
\frac{\Gamma_p \left(  \frac{k+1}{2} \right) }{ \Gamma_p\left(\frac {k}{2}\right ) }
= \frac{\Gamma \left( \frac{k+1}{2} \right) }{ \Gamma \left(\frac {k+1-p}{2}\right )  }.
\end{equation*}
Now applying the last three equations to the density (\ref{almost b density}) we get 
\begin{equation*} 
 \frac{  2^{ (p-k)/2} k^{p(k+1)/2}   \Gamma\left( \frac{k+1}{2} \right)  }{ (\pi \eta )^{p/2}    \Gamma\left(\frac {k-p+1}{2}\right )  \Gamma\left( \frac{k-p}{2} +1\right) } 
 \frac{   e^{   \mathbf{z}^{\mathsf{T}}  \boldsymbol{\lambda}/\eta}    K_{(p-k-2)/2}\left[ \left\{  \left( z^{\mathsf{T}}  \boldsymbol{\lambda}\right)^2 / \eta^2  +k   \boldsymbol{\lambda}^{\mathsf{T}}  \boldsymbol{\lambda} / \eta \right\}^{1/2} \right]  }{ 
   \left[ 1 + \frac{ \mathbf{z}^{\mathsf{T}}   \mathbf{z} }{k \eta }\right]^{(k+1)/2}
\left[ \frac{ \left( \mathbf{z}^{\mathsf{T}}   \boldsymbol{\lambda}\right)^2 }{\eta^2} + k\frac{ \boldsymbol{\lambda}^{\mathsf{T}}  \boldsymbol{\lambda}  }{\eta} \right]^{(p-k-2)/4}  } . 
\end{equation*}
Finally, performing the transformation of variables and  substitution we obtain an approximate density for the partial sketch in (\ref{approximate density unconditional}). 
\end{proof}


\section{Supplementary Summary Plots for Simulation} \label{app3}

\begin{figure}[!htbp]
\includegraphics[width=5.5in, height=2in]{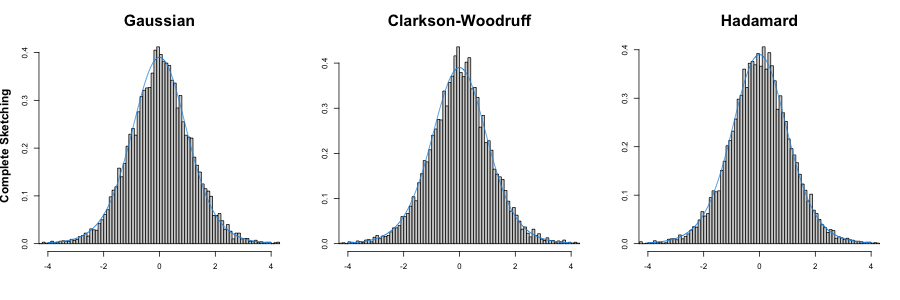}
\includegraphics[width=5.5in, height=2in]{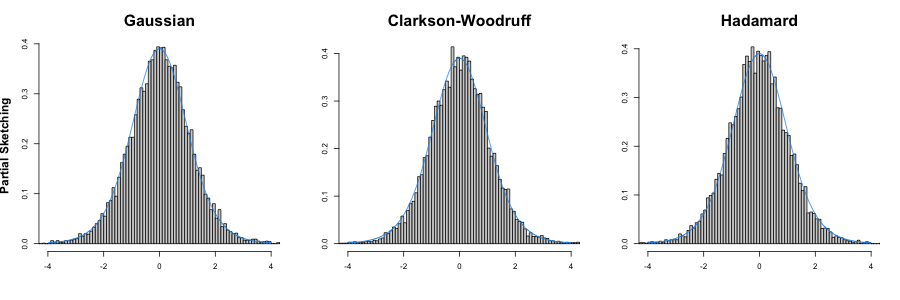}
\centering
\caption{Histograms under repeated sketches of the simulated pivotal quantity for testing if $\beta_{F6}=0$ when performing complete sketching (top row) and when performing partial sketching (bottom row) while varying the type of sketch.  The Gaussian sketch is exact and the other two are approximate. The complete and partial sketch test are based on (\ref{marginal complete test}) and (\ref{partial beta equal to zero}), respectively.}
\label{fig:testingBeta6}
\end{figure}

\begin{figure}[!htbp]
\includegraphics[width=5.5in, height=2in]{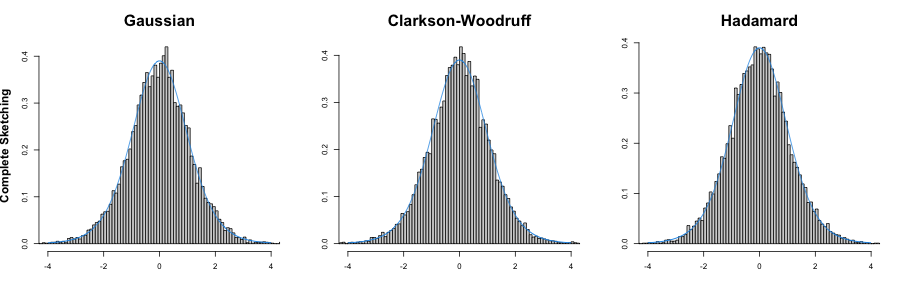}
\includegraphics[width=5.5in, height=2in]{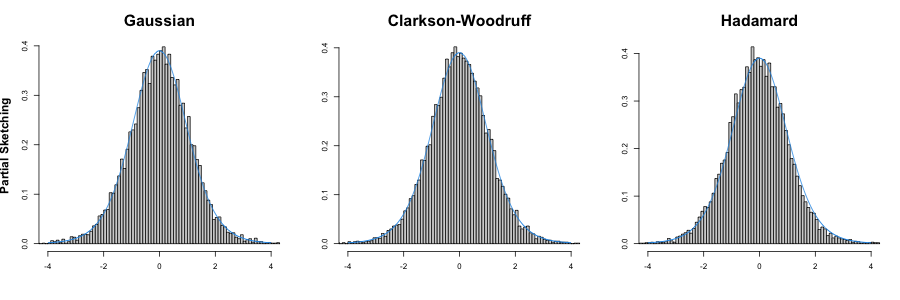}
\centering
\caption{Histograms under repeated samples of the simulated pivotal quantity for testing if $\beta_{6}=0$ when performing complete sketching (top row) and when performing partial sketching (bottom row) while varying the type of sketch.  The Gaussian sketch is exact and the other two are approximate. The complete and partial sketch test are based on (\ref{marginal complete test samples}) and (\ref{partial beta test samples}), respectively. }
\label{fig:testingBeta6samples}
\end{figure}

\bibliography{sketch}
\bibliographystyle{chicago}

\end{document}